\pdfoutput=1

\newif\ifdraft\draftfalse
\newif\iffull\fulltrue
\documentclass[a4paper,english]{lipics-v2018}
\nolinenumbers 
\usepackage{microtype}

\title{Almost Sure Productivity\footnote{%
    \iffull This is the full version of the paper.
    \else The full version of this paper can be found at
    \url{https://arxiv.org/abs/1802.06283}.
    \fi}}

    \author{Alejandro Aguirre}{IMDEA Software Institute, Madrid, Spain}{}{}{}
    \author{Gilles Barthe}{IMDEA Software Institute, Madrid, Spain}{}{}{}
    \author{Justin Hsu}{University College London, London, UK}{}{}{}
    \author{Alexandra Silva}{University College London, London, UK}{}{}{}
\authorrunning{A.\, Aguirre, G.\, Barthe, J.\, Hsu, A.\, Silva} 

\Copyright{Alejandro Aguirre, Gilles Barthe, Justin Hsu, Alexandra Silva}


\EventEditors{Ioannis Chatzigiannakis, Christos Kaklamanis, Daniel Marx, and Don Sannella}
\EventNoEds{4}
\EventLongTitle{45th International Colloquium on Automata, Languages, and Programming (ICALP 2018)}
\EventShortTitle{ICALP 2018}
\EventAcronym{ICALP}
\EventYear{2018}
\EventDate{July 9--13, 2018}
\EventLocation{Prague, Czech Republic}
\EventLogo{eatcs}
\SeriesVolume{80}
\ArticleNo{374}

\newif\ifshort
\iffull\shortfalse\else\shorttrue\fi

\usepackage{amsmath}
\usepackage{stmaryrd}
\usepackage{tikz-cd}
\tikzcdset{every label/.append style = {font = \small}}

\usetikzlibrary{shapes.geometric}

\tikzset{
    buffer/.style={
        draw,
        regular polygon,
        regular polygon sides=3
    }
}

\newcommand\triangleS[1]{\begin{tikzpicture}
\node[buffer]{#1};
\end{tikzpicture}}

\tikzstyle{output}=[circle, draw]

\usepackage[all]{xy}
\newcommand\klar[2]{\ar[#1]|-{\circ}^-{#2}}
\newcommand\klarr[2]{\ar[#1]|-{\circ}_-{#2}}

\newcommand{\EE}{\mathbb{E}}
\newcommand{\NN}{\mathbb{N}}

\newcommand{\RR}{\mathbb{R}}
\newcommand{\TT}{\mathbb{T}}

\newcommand{\Dist}{\mathcal{D}}
\newcommand{\cF}{\mathcal{F}}
\newcommand{\finally}{\lozenge}
\newcommand{\always}{\square}
\newcommand{\paths}{\mathrm{Paths}}

\newcommand{\cA}{\mathcal{A}}
\newcommand{\cB}{\mathcal{B}}
\newcommand{\cC}{\mathcal{C}}
\newcommand{\cO}{\mathcal{O}}

\newcommand{\cS}{\mathcal{S}}
\newcommand{\cT}{\mathcal{T}}
\newcommand{\nextt}[1]{\mathcal{X} #1}
\newcommand{\until}[2]{#1 \mathbin{\mathcal{U}} #2}
\newcommand{\cons}[2]{#1 : #2}

\newcommand{\sem}[1]{\llbracket #1 \rrbracket}
\newcommand{\cat}[2]{#1 \cdot #2}

\newcommand{\Trees}{\mathsf{Trees}}

\newcommand{\lt}{\mathsf{left}}
\newcommand{\rt}{\mathsf{right}}

\newcommand{\hd}{\mathsf{head}}
\newcommand{\tl}{\mathsf{tail}}
\newcommand{\Kl}{\mathcal{K}\ell}
\newcommand{\out}{\mathsf{out}}
\newcommand{\unf}{\mathsf{unf}}
\newcommand{\OS}{\mathsf{OS}}
\newcommand{\OT}{\mathsf{OT}}
\newcommand{\step}{\mathsf{st}}
\newcommand{\inl}{\mathit{inl}}
\newcommand{\inr}{\mathit{inr}}
\newcommand{\mkt}{\mathsf{mk}}
\newcommand{\trans}{\cT}
\newcommand{\subterm}[1]{\cS_{#1}}

\newenvironment{subproof}[1][\proofname]{%
  \begin{proof}[#1]%
}{%
  \end{proof}%
}

\begin{document}

\maketitle

\begin{abstract}
We introduce \emph{Almost Sure Productivity (ASP)}, a probabilistic
generalization of the productivity condition for coinductively defined
structures. Intuitively, a probabilistic coinductive stream or tree is ASP if it
produces infinitely many outputs with probability $1$. Formally, we define ASP
using a final coalgebra semantics of programs inspired by Kerstan and K\"onig.
Then, we introduce a core language for probabilistic streams and trees, and
provide two approaches to verify ASP: a syntactic sufficient criterion, and a
decision procedure by reduction to model{-}checking LTL formulas on
probabilistic pushdown automata.
\end{abstract}

\section{Introduction} \label{sec:intro}
The study of probabilistic programs has a long history, especially in connection
with semantics~\cite{DBLP:journals/jcss/Kozen81} and
verification~\cite{DBLP:conf/stoc/Kozen83,DBLP:journals/toplas/HartSP83,DBLP:journals/toplas/MorganMS96}.
Over the last decade the field of probabilistic programming has attracted
renewed attention with the emergence of practical probabilistic programming
languages and novel applications in machine learning, privacy-preserving data
mining, and modeling of complex systems. On the more theoretical side, many
semantical and syntactic tools have been developed for verifying probabilistic
properties. For instance, significant attention has been devoted to termination
of probabilistic programs, focusing on the complexity of the different
termination classes~\cite{DBLP:conf/mfcs/KaminskiK15}, and on practical methods
for proving that a program
terminates~\cite{DBLP:conf/popl/FioritiH15,DBLP:conf/esop/LagoG17,DBLP:journals/pacmpl/AgrawalC018,DBLP:journals/pacmpl/McIverMKK18}.
The latter class of works generally focuses on \emph{almost sure termination},
which guarantees that a program terminates with probability 1.

Coinductive probabilistic programming is a new computational paradigm that
extends probabilistic programming to infinite objects, such as streams and
infinite trees, providing a natural setting for programming and reasoning about
probabilistic infinite processes such as Markov chains or Markov decision
processes. Rather surprisingly, the study of coinductive probabilistic
programming was initiated only recently~\cite{DBLP:conf/esop/AguirreBBBGG18},
and little is known about generalizations of coinductive concepts and methods to
the probabilistic setting. In this paper we consider \emph{productivity}, which
informally ensures that one can compute arbitrarily precise approximations of
infinite objects in finite time. Productivity has been studied extensively for
standard, non-probabilistic coinductive
languages~\cite{DBLP:conf/popl/HughesPS96,DBLP:journals/tcs/EndrullisGHIK10,DBLP:conf/popl/AbelPTS13,DBLP:journals/corr/CloustonBGB16,DBLP:conf/lics/CaprettaF17},
but the probabilistic setting introduces new subtleties and challenges.

\paragraph*{Contributions}
Our first contribution is conceptual. We introduce \emph{almost sure
productivity} (ASP), a probabilistic counterpart to productivity.
A probabilistic stream computation is almost
surely productive if it produces an infinite stream of outputs with probability
1. For instance, consider the stream defined by the equation
$$\sigma = (a : \sigma) \oplus_p \sigma$$
Viewed as a program, this stream repeatedly flips a coin with bias $p \in (0, 1)$, producing the value $a$
if the coin comes up heads and retrying if the coin comes up tails. This
computation is almost surely productive since the probability it
fails to produce outputs for $n$ consecutive steps is $(1 - p)^n$, which tends to zero as
$n$ increases. In contrast, consider the stream
$$\sigma = \bar{a} \oplus_p \epsilon$$
This computation flips a single biased coin and returns an infinite stream of
$a$'s if the coin comes up heads, and the empty stream $\epsilon$ if the coin
comes up tails. This process is \emph{not} almost surely productive since its
probability of outputting an infinite stream is only $p$, which is strictly less
than $1$.

We define almost sure productivity for any system that can be equipped with a
final coalgebra semantics in the style of Kerstan and
K\"onig~\cite{DBLP:journals/corr/KerstanK13} (Section~\ref{sec:def-asp}). We
instantiate our semantics on a core probabilistic language for computing over
streams and trees (Section~\ref{sec:lang}). Then, we consider two methods for
proving almost sure productivity.
\begin{enumerate}
\item We begin with a syntactic method that assigns a real-valued measure to
  each expression $e$ (Section~\ref{sec:hash}). Intuitively, the measure
  represents the expected difference between the number of outputs produced and
  consumed per evaluation step of the expression. For instance, the computation
  that repeatedly flips a fair coin and outputs a value if the coin is heads has
  measure $\frac{1}{2}$---with probability $1/2$ it produces an output, with
  probability $0$ it produces no outputs. More complex terms in our language can
  also consume outputs internally, leading to possibly negative values for the
  productivity measure.

  We show that every expression whose measure is strictly positive is almost
  surely productive; the proof of soundness of the method uses concentration
  results from martingale theory. While simple to carry out, our syntactic
  method is incomplete---it does not yield any information for expressions with
  non-positive measure.
\item To give a more sophisticated analysis, we reduce the problem of deciding
  ASP to probabilistic model{-}checking (Section~\ref{sec:ppda}).  We translate
  our programs to probabilistic pushdown automata and show that almost sure
  productivity is characterized by a logical formula in LTL. This fragment is
  known to be decidable~\cite{Brazdil2013}, giving a sound and complete procedure for
  deciding ASP.
\end{enumerate}
We consider more advanced generalizations and extensions in
Section~\ref{sec:discussion}, survey related work in Section~\ref{sec:rw}, and
conclude in Section~\ref{sec:conclusion}.

\section{Mathematical Preliminaries} \label{sec:prelim}
This section reviews basic notation and definitions from measure theory and
category theory. Given a set $A$ we will denote by $A_\bot$ the coproduct of
$A$ with a one-element set containing a distinguished element $\bot$, i.e.,
$A_\bot = A + \{\bot\}$. 
\paragraph*{Coalgebra, Monads, Kleisli categories}
We assume that the reader is familiar with the notions of objects,
morphisms, functors and natural transformations (see, for instance,
\cite{awodey2006cat}). Given an endofunctor $F: {\mathcal C} \to {\mathcal C}$, a \emph{coalgebra} of $F$ is a pair $(X, f)$
of an object $X \in {\mathcal C}$ and a morphism $f : X \to F(x)$ in ${\mathcal C}$ .
A \emph{monad} is a triple $(T, \eta, \mu)$ of an endofunctor $T : {\mathcal C} \to {\mathcal C}$ and two
natural transformations $\eta: 1_{\mathcal C} \to T$ (the \emph{unit}) and
$\mu: T^2 \to T$ (the \emph{multiplication}) such that $\mu \circ \mu_T = \mu
\circ T \mu$ and $\mu \circ \eta_T = 1_T  = \mu \circ T \eta$.
Given a category ${\mathcal C}$ and a monad $(T, \eta, \mu)$, the \emph{Kleisli
category} $\Kl(T)$ of $T$ has as objects the objects of ${\mathcal C}$ and as
morphisms $X \to Y$ the morphisms $X \to T(Y)$ in ${\mathcal C}$. 
\paragraph*{Streams, Trees, Final Coalgebra} We will denote by $O^\omega$ the
set of infinite streams of elements of $O$ (alternatively characterized as
functions $\NN\to O$). We have functions $\hd\colon O^\omega \to O$ and
$\tl\colon O^\omega \to O^\omega$ that enable observation of the elements of
the stream. In fact, they provide $O^\omega$ with a one-step structure that is
canonical: given any set $S$ and any two functions $h\colon S \to O$ and
$t\colon S\to S$ (i.e., a coalgebra $(S, \langle h,t \rangle)$ of the functor
$F(X) = O \times X$) there exists a {\em unique} stream function associating
semantics to elements of $S$:
\[
\xymatrix@C=1.5cm{
S  \ar@{-->}[r]^{\sem - }\ar[d]_{<h,t>}& O^\omega\ar[d]^{<\hd, \tl>}\\
O \times S \ar@{-->}[r]^{id \times \sem - } & O \times O^\omega
}
\]
Formally, this uniqueness property is known as \emph{finality}: $O^\omega$ is the \emph{final coalgebra} of the functor $F(X) = O\times X$ and the above diagram gives rise to a coinductive definition principle. A similar principle can be obtained for infinite binary trees and other algebraic datatypes. 
The above diagrams are in the category of sets and functions, but infinite streams and trees have a very rich algebraic structure and they are also the carrier of final coalgebras in other categories. For the purpose of this paper, we will be particularly interested in a category where the maps are probabilistic---the Kleisli category of the distribution (or {\em Giry}) monad. 
\paragraph*{Probability Distributions, $\sigma$-algebras, Measurable Spaces}
To model probabilistic behavior, we need some basic concepts from measure theory
(see, e.g.,~\cite{rudin-real-complex}). Given an arbitrary set $X$ we call a set
$\Sigma$ of subsets of $X$ a $\sigma$-\emph{algebra} if it contains the empty
set and is closed under complement and countable union.  A \emph{measurable
space} is a pair $(X,\Sigma)$. A \emph{probability measure} or distribution
$\mu$ over a measurable space is a function $\mu : \Sigma \to [0,1]$ assigning
probabilities $\mu(A) \in [0,1]$ to the \emph{measurable sets} $A \in \Sigma$
such that $\mu(X) = 1$ and $\mu(\bigcup_{i \in I} A_i) = \sum_{i \in I}
\mu(A_i)$ whenever $\{A_i\}_{i \in I}$ is a countable collection of disjoint
measurable sets. 
%
The collection $\Dist(X)$ of probability distributions over a measurable space
$X$ forms the so-called Giry monad. The monad unit $\eta\colon X \to \Dist(X)$
maps $a \in X$ to the point mass (or Dirac measure) $\delta_a$, i.e., the
measure assigning $1$ to any set containing $a$ and $0$ to any set not
containing $a$. The monad multiplication $m\colon \Dist \Dist (X) \to \Dist (X)$
is given by integration:
\[ 
m(P)(S) = \int  ev_S dP , \text{ where } ev_S (\mu) = \mu(S) .
\]
Given measurable spaces $(X,\Sigma_X)$ and $(Y,\Sigma_Y)$, a \emph{Markov
kernel} is a function $P : X \times \Sigma_Y \to [0,1]$ (equivalently, $X
\to\Sigma_Y \to [0,1]$) that maps each source state $x\in X$ to a distribution
over target states $P(x,-) : \Sigma_Y \to [0,1]$.

Markov kernels form the arrows in the Kleisli category $\Kl(\Dist)$ of the
$\Dist$ monad; we denote such arrows by $\xymatrix{X \klar{r}{P}& Y}$.
Composition in the Kleisli category is given by integration:
\[
\xymatrix{X \klar{r}{P}& Y\klar{r}{Q}& Z} \qquad\qquad (P\circ Q)(x,A) = \int_{y\in Y} P(x,dy) \cdot Q(y,A)
\]
Associativity of composition is essentially Fubini's theorem.

\section{Defining Almost Sure Productivity} \label{sec:def-asp}

We will consider programs that denote probability distributions over coinductive
types, such as infinite streams or trees. In this section, we focus on the
definitions for programs producing streams and binary trees for simplicity, but
our results should extend to arbitrary polynomial functors (see
Section~\ref{sec:discussion}).

First, we introduce the semantics of programs. Rather than fix a concrete
programming language at this point, we let $\TT$ denote an abstract state space
(e.g., the terms of a programming language or the space of program memories).
The state evolves over an infinite sequence of discrete time steps. At each
step, we will probabilistic observe either a concrete output ($a \in A$) or
nothing ($\bot$), along with a resulting state. Intuitively, $p\in \TT$ is ASP
if its probability of producing unboundedly many outputs is $1$.  Formally, we
give states in $\TT$ a denotational semantics $\sem - \colon \TT \to
\Dist((A_\bot)^\omega)$ defined coinductively, starting from a given one-step
semantics function that maps each term to an output in $A_\bot$ and the
resulting term. Since the step function is probabilistic, we work in the Kleisli
category for the distribution monad; this introduces some complications when
computing the final coalgebras in this category. We take the work on
probabilistic streams by Kerstan and
K\"onig~\cite{DBLP:journals/corr/KerstanK13} as our starting point, and then
generalize to probabilistic trees. 

\begin{theorem}[Finality for streams~\cite{DBLP:journals/corr/KerstanK13}]
  Given a set $\TT$ of programs endowed with a probabilistic step function
  $\step\colon \TT \to \Dist(A_\bot \times \TT)$, there is a {\em unique}
  semantics function $\sem -$ assigning to each program a probability
  distribution of output streams such that the following diagram commutes in the
  Kleisli category $\Kl(\Dist)$:
\[
\xymatrix@C=1.5cm{
\TT  \klar{r}{\sem - }\klarr{d}{\step}& (A_\bot)^\omega\klar{d}{<\hd, \tl>}\\
A_\bot \times \TT \klar{r}{id \times \sem - } & A_\bot \times (A_\bot)^\omega
}
\]
\end{theorem}

\begin{definition}[ASP for streams] A stream program $p \in \TT$ is {\em almost surely productive} (ASP) if
\[
\underset{\sigma\sim\sem{p}}{\text{Pr}}\text{[$\sigma$ has infinitely many concrete output elements $a\in A$]} = 1. 
\]
\end{definition}

For this to be a sensible definition, the event ``$\sigma$ has infinitely many
concrete output elements $a\in A$'' must be a measurable set in some
$\sigma$-algebra on $(A_\bot)^\omega$. Following Kerstan and K\"onig, we take
the $\sigma$-algebra generated by \emph{cones}, sets of the form $uA^\omega = \{
v \in (A_\bot)^\omega \mid u\text{ prefix of } v, u\in (A_\bot)^*\}$. Our
definition evidently depends on the definition of $\sem - \colon \TT \to
\Dist(A_\bot)^\omega$; our coinductively defined semantics will be useful later
for showing soundness when verifying ASP, but our definition of ASP is sensible
for any semantics $\sem{-}$. 

\begin{example}
  Let us consider the following program defining a stream $\sigma$ recursively,
  in which each recursion step is determined by a coin flip with bias $p$:
  \[
    \sigma = (a:\sigma) \oplus_p \tl(\sigma)
  \]
  In the next section we will formally introduce this programming language, but
  intuitively the program repeatedly flips a coin. If the coin flip results in
  heads the program produces an element $a$. Otherwise the program tries to
  compute the tail of the recursive call; the first element produced by the
  recursive call is dropped (consumed), while subsequent elements produced (if
  any) are emitted as output.

  To analyze the productivity behavior of this probabilistic program, we can
  reason intuitively. Each time the second branch is chosen, the program must
  choose the first branch strictly more than once in order to produce one output
  (since, e.g., $\tl(a : \sigma) = \sigma$).  Accordingly, the productivity
  behavior of this program depends on the value of $p$. When $p$ is less than
  $1/2$, the program chooses the first branch less often than the second branch
  and the program is not ASP. On the other hand, when $p > 1/2$ the program will
  tend to produce more elements $a$ than are consumed by the destructors, and
  the above program is ASP. In the sequel, we will show two methods to formally
  prove this fact.
\end{example}

It will be convenient to represent the functor as $F(X) = A_\bot\times X$ as $A
\times X + X$. In the rest of this paper we will often use the latter
representation and refer to the final coalgebra as \emph{observation streams}
$\OS = (A_\bot)^\omega$ with structure $\xymatrix{\OS & &A \times \OS +\OS
\ar[ll]_-{<\out,\unf>}^-\cong}$ given by $\out(a,\sigma) = a:\sigma$ and
$\unf(\sigma) = \bot:\sigma$.

Streams are not the only coinductively defined data; infinite binary trees are
another classical example. To generate trees, we can imagine that a program
produces an output value---labeling the root node---and two child programs,
which then generate the left and right child of a tree of outputs. Much like we
saw for streams, probabilistic programs generating these trees may sometimes
step to a single new program without producing outputs. Accordingly we will work
with the functor $F(X) = A \times X \times X + X$, where the left summand can be
thought of as the result of an output step, while the right summand gives the
result of a non-output step.

\begin{theorem}[Finality for trees] Given a set of programs $\TT$ endowed with a probabilistic step function  $\step\colon \TT \to \Dist(A \times \TT\times \TT +\TT)$, there is a {\em unique} semantics function $\sem -$ assigning to each program a probability distribution of output trees such that the following diagram commutes in the Kleisli category $\Kl(\Dist)$. 
\[
\xymatrix@C=2.5cm{
\TT  \klar{r}{\sem - }\klarr{d}{\step}&   \Trees(A_\bot)\klar{d}{<\out, \unf>^{-1}}\\
A \times \TT\times \TT +\TT \klar{r}{id \times \sem -  \times \sem - +   \sem -  } & A \times \Trees(A_\bot) \times \Trees(A_\bot) +   \Trees(A_\bot)
}
\]
\end{theorem}
$\Trees(A_\bot)$ are infinite trees where the nodes are either elements of $A$ or $\bot$. An $a$-node has two children whereas a $\bot$-node only has one child. Formally, we can construct these trees with the two maps $\out$ and $\unf$:

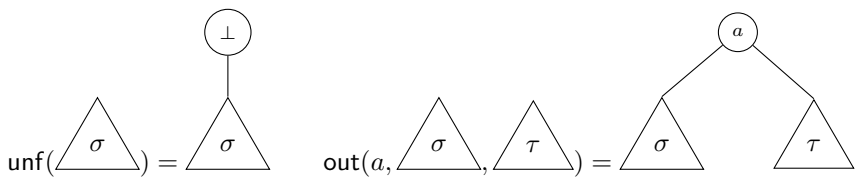
\[
  \unf (\triangleS {$\sigma$} ) =
  \begin{tikzpicture}
    \node[output](top) at (0, -0.5) {\scriptsize $\bot$};
    \node[buffer](bottom) at (0, -2) {$\sigma$};
    \draw [-] (top) -- (bottom);
  \end{tikzpicture}
  \qquad \out (a, \triangleS {$\sigma$}, \triangleS {$\tau$} ) =
  \begin{tikzpicture}
    \node[output](top) at (0, -0.5) {\scriptsize $a$};
    \node[buffer](left) at (-1, -2) {$\sigma$};
    \node[buffer](right) at (1, -2) {$\tau$};
    \draw [-] (top) -- (left.north);
    \draw [-] (top) -- (right.north);
  \end{tikzpicture}
\]

Defining ASP for trees is a bit more subtle than for streams. Due to
measurability issues, we can only refer to the probability of infinitely many
outputs along one path at a time in the tree. A bit more formally, let $w \in \{
L, R \}^\omega$ be an infinite word on alphabet $\{ L, R \}$. Given any tree $t
\in \Trees(A_\bot)$, $w$ induces a single path $t_w$ in the tree: from the root,
the path follows the left/right child of $a$-nodes as indicated by $w$, and the
single child of $\bot$-nodes.

\begin{definition}[ASP for trees] A tree program $p \in \TT$ is {\em almost surely productive} (ASP) if
\[
  \forall w \in \{ L, R \}^\omega.\,
  \underset{t\sim\sem{p}}{\text{Pr}}[t_w \text{ has infinitely many concrete output nodes } a\in A] = 1. 
\]
\end{definition}
We have omitted the $\sigma$-algebra structure on $\Trees(A_\bot)$ for lack of
space, but it is quite similar to the one for streams: it is generated by the
cones $u\Trees(A_\bot) = \{ t \in \Trees(A_\bot) \mid t \text{ is an
extension of the finite tree $u$}\}$.

\begin{example}
  Consider the probabilistic tree defined by the following equation:
  \[
    \tau = \mkt(a, \tau, \tau) \oplus_p \lt(\tau)
  \]
  The $\mkt(a, t_1, t_2)$ constructor produces a tree with the root labeled by
  $a$ and children $t_1$ and $t_2$, while the $\lt(t)$ destructor consumes the
  output at the root of $t$ and steps to the left child of $t$. While this
  example is more difficult to work out informally, it has similar ASP behavior
  as the previous example we saw for streams: when $p > 1/2$ this program is
  ASP, since it has strictly higher probability of constructing a node (and
  producing an output) than destructing a node (and consuming an output).
\end{example}

\section{A Calculus for Probabilistic Streams and Trees} \label{sec:lang}

Now that we have introduced almost sure productivity, we consider how to verify
this property. We work with two variants of a simple calculus for probabilistic
coinductive programming, for producing streams and trees respectively.  We
suppose that outputs are drawn from some finite alphabet $A$. The language for
streams considers terms of the following form:
\[
    e \in \TT ::= \sigma
    \mid e \oplus_p e
    \mid \cons{a}{e}~(a \in A)
    \mid \tl(e)
\]
The distinguished variable $\sigma$ represents a recursive occurrence of the
stream so that streams can be defined via equations $\sigma = e$.
The operation $e_1 \oplus_p e_2$ selects $e_1$
with probability $p$ and $e_2$ with probability $1 - p$.
The constructor $\cons{a}{e}$ builds a stream with head $a$ and tail $e$. The
destructor $\tl(e)$ computes the tail of a stream, discarding the head.

The language for trees is similar, with terms of the following form:
\[
    e \in \TT ::= \tau
    \mid e \oplus_p e
    \mid \mkt(a,e,e)~(a \in A)
    \mid \lt(e)
    \mid \rt(e)
\]
The variable $\tau$ represents a recursive occurrence of the tree, so that trees
are defined as $\tau = e$. The constructor $\mkt(a,e_1,e_2)$ builds a tree with
root labeled $a$ and children $e_1$ and $e_2$. The destructors $\lt(e)$ and
$\rt(e)$ extract the left and right children of $e$, respectively.

We interpret these terms coalgebraically by first giving a step function from
$\step_e : \TT \to \Dist(F(\TT))$ for an appropriate functor, and then taking
the semantics as the map to the final coalgebra. For streams, we
take the functor $F(X) = A \times X + X$: a term steps to a distribution over
either an output in $A$ and a resulting term, or just a resulting term (with no
output). To describe how the recursive occurrence $\sigma$ steps, we parametrize
the step function $\step_e$ by the top level stream term $e$; this term remains
fixed throughout the evaluation. This choice restricts recursion to be global
in nature, i.e., our language does not support mutual or nested recursion.
Supporting more advanced recursion is also possible, but we stick with the
simpler setting here; we return to this point in Section~\ref{sec:discussion}.

The step relation is defined by case analysis on the syntax of terms.
Probabilistic choice terms reduce by scaling the result of stepping $e$ and
the result of stepping $e'$ by $p$ and $1-p$ respectively, and then combining
the distributions:
\begin{align*}
  \step_e(e_1 \oplus_p e_2) &\triangleq p \cdot \step_e(e_1) + (1 - p) \cdot \step_e(e_2)
\end{align*}
%
The next cases push destructors into terms:
\begin{align*}
  \step_e(\tl^k(\cons{a}{e})) &\triangleq \step_e(\tl^{k - 1}(e)) \\
  \step_e(\tl^k(e_1 \oplus_p e_2)) &\triangleq \step_e(\tl^k(e_1) \oplus_p \tl^k(e_2))
\end{align*}
%
Here and below, we write $\tl^k$ as a shorthand for $k > 0$ applications of
$\tl$.

The remaining cases return point distributions. If we have reached a constructor
then we produce a single output. Otherwise, we replace $\sigma$ by the top level
stream term, unfolding a recursive occurrence.
\begin{align*}
  \step_e(\cons{a}{e'}) &\triangleq \delta(\inl(a, e')) \\
  \step_e(e') &\triangleq \delta(\inr(e'[e / \sigma] ))
  \quad \text{otherwise}
\end{align*}
Note that a single evaluation step of a stream may lead to multiple constructors
at top level of the term, but only one output can be recorded each step---the
remaining constructors are preserved in the term and will give rise to outputs
in subsequent steps.

The semantics is similar for trees. We take the functor $F(X) = (A \times X
\times X) + X$: a term reduces to a distribution over either an output in $A$
and two child terms, or a resulting term and no output. The main changes to the
step relation are for constructors and destructors. The constructor $\mkt(a,
e_1, e_2)$ reduces to $\delta(\inl(a,e_1,e_2))$, representing an output $a$ this
step. Destructors are handle like $\tl$ for streams, where $\lt(\mkt(a,
e_1, e_2))$ reduces to $e_1$ and $\rt(\mkt(a, e_1, e_2))$ reduces to $e_2$, and
$\tl^k(-)$ is generalized to any finite combination of $\lt(-)$ and $\rt(-)$.

Concretely, let $C[e]$ be any (possibly empty) combination of $\lt$ and $\rt$
applied to $e$. We have the following step rules:
\begin{align*}
  \step_e(C[\lt(\mkt(a, e_l, e_r))]) &\triangleq \step_e(C[e_l]) \\
  \step_e(C[\rt(\mkt(a, e_l, e_r))]) &\triangleq \step_e(C[e_r]) \\
  \step_e(C[e_1 \oplus_p e_2]) &\triangleq p \cdot \step_e(C[e_1]) + (1 - p) \cdot \step_e(C[e_2]) \\
  \step_e(\mkt(a, e_l, e_r)) &\triangleq \delta(\inl(a, e_l, e_r)) \\
  \step_e(C[\tau]) &\triangleq \delta(\mathit{\inr}(C[e]))
\end{align*}
%
%
%
%

\section{Syntactic Conditions for ASP} \label{sec:hash}

With the language and semantics in hand, we now turn to proving ASP.  While it
is theoretically possible to reason directly on the semantics using our
definitions from Section~\ref{sec:def-asp}, in practice it is much easier to
reason about the language. In this section we present a syntactic sufficient
condition for ASP. Intuitively, the idea is to approximate the expected number
of outputs every step; if this measure is strictly positive, then the program is
ASP.

\subsection{A Syntactic Measure}

We define a syntactic measure $\#(-) : \TT \to \RR$ by induction on stream terms:
\begin{align*}
  \#(\sigma) &\triangleq 0 \\
  \#(e_1 \oplus_p e_2) &\triangleq p \cdot \#(e_1) + (1 - p) \cdot \#(e_2) \\
  \#(\cons{a}{e}) &\triangleq \#(e) + 1 \\
  \#(\tl(e)) &\triangleq \#(e) - 1
\end{align*}
The measure $\#$ describes the expected difference between the number of outputs
produced (by constructors) and the number of outputs consumed (by destructors)
in each unfolding of the term.
We can define a similar measure for tree terms:
\begin{align*}
  \#(\tau) &\triangleq 0 \\
  \#(e_1 \oplus_p e_2) &\triangleq p \cdot \#(e_1) + (1 - p) \cdot \#(e_2) \\
  \#(\mkt(a, e_1, e_2)) &\triangleq \min(\#(e_1), \#(e_2)) + 1 \\
  \#(\lt(e)) = \#(\rt(e)) &\triangleq \#(e) - 1
\end{align*}
We can now state conditions for ASP for streams and trees.

\begin{theorem} \label{thm:hash:stream}
  Let $e$ be a stream term with $\gamma = \#(e)$. If $\gamma > 0$, $e$ is ASP.
  
\end{theorem}

%


\begin{theorem} \label{thm:hash:tree}
  Let $e$ be a tree term with $\gamma = \#(e)$.  If $\gamma > 0$, $e$ is ASP.
\end{theorem}
%

\iffull
\subsection{Soundness}
\fi

The main idea behind the proof for streams is that by construction of the step
relation, each step either produces an output or unfolds a fixed point (if there
is no output). In unfolding steps, the expected measure of the term plus the
number of outputs increases by $\gamma$. By defining an appropriate martingale
and applying the Azuma-Hoeffding inequality, the sum of the measure and the
number of outputs must increase linearly as the term steps when $\gamma > 0$.
Since the measure is bounded above---when the measure is large the stream
outputs instead of unfolding---the number of outputs must increase linearly and
the stream is ASP.

\iffull
We will need a few standard constructions and results from probability theory.

\begin{definition}[See, e.g.,~\cite{durrett2010probability}]
  A \emph{filtration} $\{ \cF_i \}_{i \in \NN}$ of a $\sigma$-algebra $\cF$ on a
  measurable space $A$ is an sequence of $\sigma$-algebras such that $\cF_i
  \subseteq \cF_{i + 1}$ and $\cF_i \subseteq \cF$, for all $i \in \NN$. A
  \emph{stochastic process} is a sequence of random variables $\{ X_i : A \to B
  \}_{i \in \NN}$ for $B$ some measurable space, and the process is
  \emph{adapted to the filtration} if every $X_i$ is $\cF_i$-measurable.
\end{definition}

Intuitively, a filtration gives each event a time $i$ at which the event starts
to have a well-defined probability. A stochastic process is adapted to the
filtration if its value at time $i$ only depends on events that are well-defined
at time $i$ or before (and not events at future times).

An important class of stochastic processes are martingales.

\begin{definition}[See, e.g.,~\cite{durrett2010probability}]
  Let $\{ X_i : A \to \RR \}$ be a real-valued stochastic process adapted to
  some filtration on $A$, and let $\mu$ be a measure on $A$. Suppose that
  $\EE_\mu [ |X_i| ] < \infty$ for all $X_i$. The sequence is a \emph{martingale}
  if for all $i \in \NN$, we have
  \[
    \EE_{\mu} [ X_{i + 1} \mid \cF_i ] = X_i .
  \]
  The conditional expectation turns $X_{i + 1}$ from an $\cF_{i + 1}$-measurable
  map to an $\cF_i$-measurable map; equivalently, the martingale condition can
  be stated as
  \[
    \EE_{\mu} [ (X_{i + 1} - X_i) \chi_{F} ] = 0 ,
  \]
  for every event $F \in \cF_i$, where $\chi_{F}$ is the indicator function.
  If the equalities are replaced by $\geq$ (resp., $\leq$), then the sequence is
  a sub- (resp., super-) martingale. 
\end{definition}

Martingale processes determine a sequence of random variables that may not be
independent, but where the expected value of the process at some time step
depends only on its value at the previous time step. Martingales satisfy
concentration inequalities.

\begin{theorem}[Azuma-Hoeffding inequality~\cite{azuma1967weighted}] \label{thm:azuma}
  Let $\{ X_i \}_i$ be sequence such that $|X_{i + 1} - X_i| \leq
  c$ for all $i \in \NN$. If $\{ X_i \}_i$ is a sub-martingale, then for every
  $n \in \NN$ and $B \geq 0$, we have
  \[
    \Pr [ X_n - X_0 \leq -B ] \leq \exp( - B^2 / 2 n c ) .
  \]
  If $\{ X_i \}_i$ is a super-martingale, then for every
  $n \in \NN$ and $B \geq 0$, we have
  \[
    \Pr [ X_n - X_0 \geq B ] \leq \exp( - B^2 / 2 n c ) .
  \]
  If $\{ X_i \}_i$ is both a martingale, then combining the above results gives
  \[
    \Pr [ |X_n - X_0| \geq B ] \leq 2 \exp( - B^2 / 2 n c ) .
  \]
\end{theorem}

\begin{proof}[Proof of Theorem~\ref{thm:hash:stream}]
  While the semantics constructed in Section~\ref{sec:def-asp} is
  sufficient to describe ASP, for showing soundness it is more convenient to
  work with an instrumented semantics that tracks the term in the observation stream.
  We can give a step function of type $\step_e' : \TT \to \Dist(F'(\TT))$, where $F'(X) =
  (A \times X + X) \times \TT$ by recording the input term in the output. For
  instance:
  \begin{align*}
    \step_e'(\cons{a}{e'}) &\triangleq \delta(\inl(a, e'), \cons{a}{e'}) \\
    \step_e'(\sigma) &\triangleq \delta(\inr(e), \sigma)
  \end{align*}
  and so forth. Using essentially the same construction as in
  Section~\ref{sec:def-asp}, we get an instrumented semantics $\sem{-}' : \TT
  \to \Dist(\OS')$, where $\OS'$ are infinite streams with constructors $\out' : (A
  \times \OS') \times \TT \to \OS'$ and $\unf' : \OS' \times \TT \to \OS'$,
  representing output and unfold steps respectively. Letting the map $u : \OS'
  \to \OS$ simply drop the instrumented terms, the map $\Dist(u) \circ \sem{-}' : \TT
  \to \Dist(\OS)$ coincides with the semantics $\sem{-}$ defined in
  Section~\ref{sec:def-asp} by finality.

  Now, we define a few stochastic processes. Let $\{ T_i : \OS' \to \TT \}_i$ be
  the sequence of instrumented terms with $T_0 = e$, $\{ O_i : \OS' \to \{ 0, 1
  \} \}_i$ be $1$ if the $i$th node is an output node and $0$ if not, $\{ U_i :
\OS' \to \{ 0, 1 \} \}_i = \{ 1 - O_i \}_i$. It is straightforward to show that
$T_i$ is $\cF_{i - 1}$-measurable (and hence $\cF_i$-measurable), and $O_i, U_i$
are $\cF_i$-measurable---all three processes are defined by the events in the
first $i$ steps.

  Now for any stream term $t \in \TT$, we claim that
  \[
    \EE_{\step_e'(t)} [ (\inl(-, t'), -) \to 1 + \#(t') \text{ else } (\inr(t'), -) \to \#(t') - \gamma ] = \#(t) .
  \]
  This follows by induction on terms using the definition of $\step_e'$. We can
  lift the equality to the semantics, giving
  \[
    \EE_{\sem{T_0}'} [ O_{i + 1} - \gamma U_{i + 1} + \#(T_{i + 2}) \mid \cF_i ] = \#(T_{i + 1})
  \]
  noting that $T_{i + 1}$ is $\cF_i$-measurable and recalling that $T_0 = e$ is
  the initial term. We now define another stochastic process via
  \[
    X_i \triangleq \sum_{j = 0}^{i} O_j - \gamma \sum_{j = 0}^{i} U_j + \#(T_{i + 1}) .
  \]
  Note that $X_{i}$ is $\cF_i$-measurable. As we will show, this process tends
  towards zero, the second term decreases, and the third term remains bounded.
  Hence, the first term---the cumulative number of outputs---must tend towards
  infinity. Evidently each $X_i$ is bounded and we are working with probability
  measures, so each $X_i$ is integrable. We can directly check that $\{ X_i
  \}_i$ is a martingale:
  \begin{align*}
    \EE_{\sem{T_0}'}
    [ X_{i + 1} \mid \cF_i ] 
    &= \EE_{\sem{T_0}'}
    \left[ \sum_{j = 0}^{i} O_j + O_{i + 1}
      - \gamma \sum_{j = 0}^{i} U_j - \gamma U_{i + 1} + \#(T_{i + 2}) \mid \cF_i \right] \\
    &= \EE_{\sem{T_0}'}
      \left[ \sum_{j = 0}^{i} O_j - \gamma \sum_{j = 0}^{i} U_j + \#(T_{i + 1}) \mid \cF_i \right] \\
    &= \sum_{j = 0}^{i} O_j - \gamma \sum_{j = 0}^{i} U_j + \#(T_{i + 1}) = X_{i} .
  \end{align*}
  We now claim that $\#(T_i) \leq c'$ where $c'$ is one more than the number of
  constructors in the original term $T_0$. This follows by observing that (i)
  the step function increases the measure by at most the number of constructors
  or $1$ every step, and (ii) the step function only unfolds if a term reduces
  to a term with non-positive measure. Similarly,
  \[
    \sum_{j = 0}^{i} U_j \geq \lfloor i / c' \rfloor
  \]
  since each unfolding step leads to at most $c'$ output (non-unfolding) steps.

  Since $O_i$ and $U_i$ are both in $\{ 0, 1 \}$, this implies that $|X_{i + 1}
  - X_i|$ is bounded by some constant $c = c' + 2$, depending only on the
  initial term. We can now apply the Azuma-Hoeffding inequality
  (Theorem~\ref{thm:azuma}). For every $n \in \NN$ and $B \geq 0$, we have
  \[
    \Pr_{\sem{T_0}'} [ X_n - X_0 \geq -B ] \geq 1 - \exp(-B^2 / 2 n c) .
  \]
  Taking $B = n^{2/3}$, we have
  \[
    \Pr_{\sem{T_0}'} [ X_n \geq X_0 - n^{2/3} ] \geq 1 - \exp(- n^{1/3} / 2 c) .
  \]
  We also know that the total number of outputs is at least
  \[
    \sum_{j = 0}^{n} O_j = X_n + \gamma \sum_{j = 0}^{n} U_j  - \#(T_{n + 1})
    \geq X_n + \gamma \lfloor n / c' \rfloor .
  \]
  So if $\gamma > 0$, the stream has zero probability of producing at most $M$
  outputs for any finite $M$. This is because for $X_n$ is at least $-n^{2/3}$
  with probability arbitrarily close to $1$ (for large enough $n$), and $\gamma
  \lfloor n / c' \rfloor$ grows linearly in $n$ for $\gamma$ positive.
  Hence, the term is ASP.
\end{proof}
\fi

The proof for trees is similar, showing that on any path through the observation tree
there are infinitely many output steps with probability $1$.
\ifshort
We present detailed proofs in the full version of this paper.
\else
\begin{proof}[Proof of Theorem~\ref{thm:hash:tree}]
  We again work with an instrumented semantics based on the step
  function $\step_e' : \TT \to \Dist(F'(\TT))$, where $F'(X) = (A \times X \times X + X) \times
  \TT$ by recording the input term in the output. For instance:
  \begin{align*}
    \step_e'(\mkt(a, e_1, e_2)) &\triangleq \delta(\inl(a, e_1, e_2), \mkt(a, e_1, e_2)) \\
    \step_e'(\tau) &\triangleq \delta(\inr(e), \tau)
  \end{align*}
  and so forth. Using essentially the same construction as in
  Section~\ref{sec:def-asp}, we get an instrumented semantics $\sem{-}' : \TT
  \to \Dist(\OT')$, where $\OT'$ are infinite trees with constructors $\out' : (A
  \times \OT' \times \OT') \times \TT \to \OT'$ and $\unf' : \OT' \times \TT \to
  \OT'$, representing output and unfold steps respectively. Letting the map $u :
  \OT' \to \OT$ simply drop the instrumented terms, the map $\Dist(u) \circ \sem{-}' :
  \TT \to \Dist(\OT)$ coincides with the semantics $\sem{-}$ defined in
  Section~\ref{sec:def-asp} by finality.

  Let $w \in \{ L, R \}^\omega$ be any infinite word, describing whether to
  follow the left or right child of a tree. Each word determines a path through
  an observation tree: on unfold nodes we simply follow the child, while on output
  nodes we follow the child indicated by $w$. We aim to show that if $\gamma >
  0$, then there are infinitely many output nodes along this path with
  probability $1$. If this holds for all $w$, then the tree term must be
  ASP.

  To model the path, we define a sequence $\{ P_i : \OT' \to A \times \TT
  \times \TT + \TT \}_i$ inductively.  $P_0$ is simply the root of the output
  tree $\OT'$. Given $P_0, \dots, P_i$, we define $P_{i + 1}$ to be a child of
  $P_i$ as follows. If $P_i$ is an unfold node it only has one child, so we take
  $P_{i + 1}$ to be this child. Otherwise we take $P_{i + 1}$ to be the child of
  $P_i$ indicated by $w_{j + 1}$, where $j$ is the number of output nodes in
  $P_0, \dots, P_i$. The process $\{ P_i \}_i$ is adapted to the filtration on
  $\OT'$.  (Note that all indices start at $0$.)
  
  Now, we can define similar processes as in the stream case with respect to the
  path.  Let $\{ T_i : \OS' \to \TT \}_i$ be the sequence of instrumented terms
  along the path with $T_0 = e$, $\{ O_i : \OS' \to \{ 0, 1 \} \}_i$ be
  $1$ if $P_i$ is an output node and $0$ if not, $\{ U_i : \OS' \to \{ 0, 1 \}
  \}_i = \{ 1 - O_i \}_i$. It is straightforward to show that $T_i$ is $\cF_{i - 1}$-measurable
  (and hence $\cF_i$-measurable), and $O_i, U_i$ are $\cF_i$-measurable---all
  three processes are defined by the events in the first $i$ steps.

  Now for any tree term $t \in \TT$, we have
  \[
    \EE_{\step_e'(t)} [ (\inl(-, t_1, t_2), -) \to 1 + \min(\#(t_1), \#(t_2)) \text{ else } (\inr(t'), -) \to \#(t') - \gamma ] \geq \#(t)
  \]
  by induction on terms using the definition of $\step_e'$. The inequality arises
  from applying a destructor to a constructor---we may end up with a child term
  that has larger measure than the parent term, since the measure of a
  constructor takes the smaller measure of its children. We can lift the
  inequality to the semantics, giving
  \[
    \EE_{\sem{T_0}'} [ O_{i + 1} - \gamma U_{i + 1} + \#(T_{i + 2}) \mid \cF_i ] \geq \#(T_{i + 1})
  \]
  noting that $T_{i + 1}$ is $\cF_i$-measurable and letting $T_0 = e$ be the
  initial term. We can now our invariant process
  \[
    X_i \triangleq \sum_{j = 0}^{i} O_j - \gamma \sum_{j = 0}^{i} U_j + \#(T_{i + 1})
  \]
  which is a sub-martingale:
  \begin{align*}
    \EE_{\sem{T_0}'}
    [ X_{i + 1} \mid \cF_i ] 
    &= \EE_{\sem{T_0}'}
    \left[ \sum_{j = 0}^{i} O_j + O_{i + 1}
      - \gamma \sum_{j = 0}^{i} U_j - \gamma U_{i + 1} + \#(T_{i + 2}) \mid \cF_i \right] \\
    &\geq \EE_{\sem{T_0}'}
      \left[ \sum_{j = 0}^{i} O_j - \gamma \sum_{j = 0}^{i} U_j + \#(T_{i + 1}) \mid \cF_i \right] \\
    &= \sum_{j = 0}^{i} O_j - \gamma \sum_{j = 0}^{i} U_j + \#(T_{i + 1}) = X_{i} .
  \end{align*}
  The remainder of the proof is now quite similar to the stream case. $\#(T_i)
  \leq c'$ where $c'$ is one more than the number of constructors in the
  original term $T_0$. This follows by observing that (i) the step function
  increases the measure by at most the number of constructors or $1$ every
  unfolding step, and (ii) the step function only unfolds if a term reduces to a
  term with non-positive measure. Similarly,
  \[
    \sum_{j = 0}^{i} U_j \geq \lfloor i / c' \rfloor
  \]
  since each unfolding step leads to at most $c'$ output (non-unfolding) steps.

  Since $O_i$ and $U_i$ are both in $\{ 0, 1 \}$, this implies that $|X_{i + 1}
  - X_i|$ is bounded by some constant $c = c' + 2$, depending only on the
  initial term. We can now apply the Azuma-Hoeffding inequality
  (Theorem~\ref{thm:azuma}). For every $n \in \NN$ and $B \geq 0$, we have
  \[
    \Pr_{\sem{T_0}'} [ X_n - X_0 \geq -B ] \geq 1 - \exp(-B^2 / 2 n c) .
  \]
  Taking $B = n^{2/3}$, we have
  \[
    \Pr_{\sem{T_0}'} [ X_n \geq X_0 - n^{2/3} ] \geq 1 - \exp(- n^{1/3} / 2 c) .
  \]
  We also know that the total number of outputs along the path $w$ is at least
  \[
    \sum_{j = 0}^{n} O_j = X_n + \gamma \sum_{j = 0}^{n} U_j  - \#(T_{n + 1})
    \geq X_n + \gamma \lfloor n / c' \rfloor .
  \]
  So if $\gamma > 0$, the stream has zero probability of producing at most $M$
  outputs along $w$ for any finite $M$. This is because for $X_n$ is at least
  $-n^{2/3}$ with probability arbitrarily close to $1$ (for large enough $n$),
  and $\gamma \lfloor n / c' \rfloor$ is growing linearly in $n$ for $\gamma$
  positive. Since the tree term produces at least $M$ outputs along path $w$
  with probability $1$ for every $M$ and every $w$, it is ASP.
\end{proof}
\fi

\subsection{Examples}

We consider a few examples of our analysis. The alphabet $A$ does not affect the
ASP property; without loss of generality, we can let the alphabet $A$ be the
singleton $\{ \star \}$.

\begin{example}
  Consider the stream definition $\sigma = (\cons{\star}{\sigma}) \oplus_p
  \tl(\sigma)$. The $\#$ measure of the stream term is $p \cdot 1 + (1 - p)
  \cdot (-1) = 2p - 1$. By Theorem~\ref{thm:hash:stream}, the stream is ASP when
  $p > 1/2$. 
\end{example}

The measure does not give useful information when $\#$ is not positive.

\begin{example}
  Consider the stream definition $\sigma = (\cons{\star}{\sigma}) \oplus_{1/2}
  \tl(\sigma)$; the $\#$ measure of the term is $0$.  The number of outputs can
  be modeled by a simple random walk on a line, where the maximum position is
  the number of outputs produced by the stream. Since a simple random walk has
  probability $1$ of reaching every $n \in \NN$~\cite{LevinPW09}, the stream term is
  ASP.

  In contrast, $\#(\sigma)$ is $0$ but the stream definition $\sigma = \sigma$
  is clearly non-productive.
\end{example}

We can give similar examples for tree terms.

\begin{example}
  Consider the tree definitions $\tau = e_i$, where
  \begin{itemize}
    \item $e_1 \triangleq \lt(\tau) \oplus_{1/4} \mkt(\star, \tau, \tau)$
    \item $e_2 \triangleq \lt(\tau) \oplus_{1/4} \mkt(\star, \tau, \lt(\tau))$ .
  \end{itemize}
  We apply Theorem~\ref{thm:hash:tree} to deduce ASP.  We have $\#(e_1) = (1/4)
  \cdot (-1) + (3/4) \cdot (+1) = 1/2$, so the first term is ASP.  For the
  second term, $\#(e_2) = (1/4) \cdot (-1) + (3/4) \cdot 0 = -1/4$, so our
  analysis does not give any information.
  
\end{example}

\section{Probabilistic Model{-}Checking for ASP} \label{sec:ppda}

The syntactic analysis for ASP is simple, but it is not complete---no
information is given if the measure is not positive. In this section we give a
more sophisticated, complete analysis by first modeling the operational
semantics of a term by a \emph{Probabilistic Pushdown Automaton} (pPDA), then deciding
ASP by reduction to model{-}checking.

\subsection{Probabilistic Pushdown Automata and LTL}

A pPDA is a tuple $\cA = (S, \Gamma, \trans)$ where $S$ is a finite set of
states and $\Gamma$ is a finite \emph{stack alphabet}.  The \emph{transition
function} $\trans : S \times (\Gamma \cup \{\bot\}) \times S \times
\Gamma^{\ast} \to [0,1]$ looks at the top symbol on the stack (which might be
empty, denoted $\bot$), consumes it, and pushes a (possibly empty, denoted
$\varepsilon$) string of symbols onto the stack, before transitioning to the
next state. A \emph{configuration} of $\cA$ is an element of $\cC = S \times
\Gamma^\ast$, and represents the state of the pPDA and the contents of its stack
(with the top on the left) at some point of its execution. Given a
configuration, the transition function $\trans$ specifies a distribution over
configurations in the next step. Given an initial state $s$ and an initial stack
$\gamma \in \Gamma^*$, $\trans$ induces a distribution $\paths(s, \gamma)$ over
the infinite sequence of configurations starting in $(s, \gamma)$.

Linear Temporal Logic (LTL) \cite{Pnueli1977} is a linear-time temporal logic that
describes runs of a transition system, which in a pPDA correspond to infinite
sequences of configurations. Propositions in LTL are defined by the syntax
\[
  \phi, \psi ::= Q \mid \neg\phi \mid \nextt{\phi} \mid \until{\phi}{\psi} \mid \finally \phi \mid \always \phi
\]
where $\phi,\psi$ are \emph{path formulas}, which describe a particular path,
and $Q$ is a set of atomic propositions. The validity of an LTL formula on a
run $\pi$ of pPDA $\cA$ is defined as follows:
\[
\begin{aligned}
 \pi \models \Phi &\Leftrightarrow \pi[0] \in Q\\
 \pi \models \neg \phi &\Leftrightarrow \pi \not\models \phi\\
 \pi \models \nextt{\phi} &\Leftrightarrow \pi_1 \models \phi
\end{aligned}
\qquad\qquad\qquad
\begin{aligned}
 \pi \models \until{\phi}{\psi} &\Leftrightarrow \exists i. \pi_i \models \psi \wedge \forall j<i. \pi_j \models \phi \\
 \pi \models \finally{\phi} &\Leftrightarrow \exists i. \pi_i \models \phi \\
 \pi \models \always{\phi} &\Leftrightarrow \forall i. \pi_i \models \phi 
\end{aligned}
\]
where atomic propositions $q$ are interpreted as $\sem{q} \subseteq \cC$.  As
expected, path formulas are interpreted in traces of configurations $\pi \in
\cC^\omega$; $\pi[i]$ is the $i$th element in the path $\pi$, and $\pi_i$ is
the suffix of $\pi$ from $\pi[i]$.  


Given a pPDA $\cA$, a starting configuration $(s, \gamma) \in \cC$ and a LTL
formula $\phi$, the \emph{qualitative model{-}checking problem} is to decide
whether runs starting from $(s, \gamma)$ satisfy $\phi$ almost surely, i.e.,
whether $\Pr_{\pi \in \paths(s,\gamma)}[\pi \models \phi]=1$.  The following is
known.

\begin{theorem}[Br{\'a}zdil, et al.~\cite{Brazdil2013}]
  \label{thm:mc-dec}
  The quantitative model{-}checking problem for pPDAs against LTL specifications is decidable.
\end{theorem}

Almost sure productivity states that an event---namely, producing an
output---occurs infinitely often with probability 1. Such properties can be
expressed in LTL.

\begin{lemma}
  Let $(s, \gamma) \in \cC$ be an initial configuration and $\cB \subseteq \cC$
  be a set of configurations. Then 
  $\Pr_{\pi \in {\paths}(s, \gamma)}[\pi \text{ visits } \cB \text{ infinitely often}] = 1$
  iff $\Pr_{\pi \in {\paths}(s, \gamma)}[\pi \models\always\finally\cB] = 1$.
\end{lemma}
\iffull
\begin{proof}
  Trivially by computing the semantics.
\end{proof}
\fi

We will encode language terms as pPDAs and cast almost sure productivity as an
LTL property stating that configurations representing output steps are reached
infinitely often with probability $1$. Theorem~\ref{thm:mc-dec} then gives a
decision procedure for ASP. In general, this algorithm\footnote{%
  Technically, this algorithm requires first encoding the LTL
  formula into a Deterministic Rabin Automaton (DRA). Even though this encoding
  can in general blow up the problem size exponentially, this is not the case
  for the simple conditions we consider.}
is in {\bf PSPACE}.

\subsection{Modeling streams with pPDAs}

The idea behind our encoding from terms to pPDAs is simple to describe. The
states of the pPDA will represent subterms of the original term, and transitions
will model steps. In the original step relation, the only way a subterm can step
to a non-subterm is by accumulating destructors. We use a single-letter stack
alphabet to track the number of destructors so that a term like $\tl^k(e)$ can
be modeled by the state corresponding to $e$ and $k$ counters on the stack. More
formally, given a stream term $e$ we define a pPDA $\cA_{e} = (\subterm{e}, \{
tl \}, \trans_{e})$, where $\subterm{e}$ is the set of syntactic subterms of $e$
and $\trans_e$ is the following transition function:
\[
\begin{aligned}
  \trans_e((\sigma, a), (e, a)) &= 1 \\
  \trans_e((e_1 \oplus_p e_2, a),(e_1,a)) &= p \\
  \trans_e((e_1 \oplus_p e_2, a),(e_2,a)) &= 1-p
\end{aligned}
\qquad
\begin{aligned}
  \trans_e((\cons{a'}{e'}, \bot),(e',\varepsilon)) &= 1 \\
  \trans_e((\cons{a'}{e'}, tl),(e', \varepsilon)) &= 1 \\
  \trans_e((\tl(e'), a),(e', \cat{tl}{a}) &= 1
\end{aligned}
\]
Above, $\cat{}{}$ concatenates strings and we implicitly treat $a$ as alphabet
symbol or a singleton string.  All non-specified transitions have zero
probability.  We define the set of \emph{outputting configurations} as $\cO
\triangleq \{ s \in \cC \:|\: \exists a',e'.\, s = (\cons{a'}{e'}, \bot) \}$, that
is, configurations where the current term is a constructor and there are no
pending destructors. Our main result states that this set is visited infinitely
often with probability 1 if and only if $e$ is ASP. In fact, we prove something
stronger:

\begin{theorem}\label{thm:ASP-streams-dec}
  Let $e$ be a stream term and let $\cA_{e}$ be the corresponding pPDA.  Then,
  \[
    \Pr_{t \sim \sem{e}} [ t \text{ has infinitely many output nodes} ]
    = \Pr_{\pi \sim \paths(e, \varepsilon)}[\pi \models \always\finally\: \cO] .
  \]
  In particular, $e$ is ASP if and only if for almost all runs $\pi$ starting in $(e, \varepsilon)$,
  $\pi \models \always\: \finally\: \cO$.
\end{theorem}

\iffull
\begin{proof}

The first part of the proof consists on simplifying the automaton $\cA_{e}$ into
a new automaton with a transition function $\trans'$ that is synchronized to the step
function for streams considered in Section~\ref{sec:hash}, while preserving the
validity of $\pi \models \always\finally\: \cO$ for every path $\pi$.  The
simplified transition function needs to skip over all states of the form
$\tl(e')$ or $e_1 \oplus e_2$ until it reaches a state of the form
$\cons{a}{e'}$ or $\sigma$. We proceed in two steps.
\begin{enumerate}
  \item For every state of the form $\tl^{k+1}(e)$ such that $e$ does not have a $tl$
    on top, we add a new transition from $\tl^{k+1}(e)$ to $e$ so that
    $\trans'((\tl^{k+1}(e), a), (e, \cat{tl^{k+1}}{a})) = 1$ and we remove the
    transition from $\tl^{k+1}(e)$ to $\tl^k(e)$. Then we remove unreachable states.
  
  \item States of the form $e_1 \oplus_p e_2$ are removed, and for every
    transition from some $e$ to $e_1 \oplus_p e_2$ such that
    $\trans((e,a), (e_1 \oplus_p e_2, \gamma)) = q > 0 $, we add new
    transitions from $e$ to $e_1$ and $e_2$ so that $\trans'((e,a), (e_1,
    \gamma)) = p q$ and $\trans'((e,a), (e_2, \gamma)) = (1-p) q$.
\end{enumerate}

Notice that this step can be removed if we construct a reduced pPDA from the beginning,
the choice of the given construction is motivated by clarity.

Now, the transition function induces a map $\overline{\trans} : \cC \to
\Dist(A \times \cC + \cC)$ from configurations to output distributions over
configurations and outputs from one step of the pPDA:
\begin{align*}
 \overline{\trans}(\sigma, \gamma)
 &\triangleq \delta(\inr(e, \gamma)) \\
 \overline{\trans}(e_1 \oplus_p e_2, \gamma)
 &\triangleq p \cdot \overline{\trans}(e_1, \gamma) + (1-p) \cdot \overline{\trans}(e_2, \gamma) \\
 \overline{\trans}(\cons{a}{e'}, \varepsilon)
 &\triangleq \delta(\inl(u, (e', \varepsilon))) \\
 \overline{\trans}(\cons{a}{e'}, \cat{tl}{\gamma})
 &\triangleq \overline{\trans} (e', \gamma) \\
 \overline{\trans}(\tl(e'), \gamma)
 &\triangleq \overline{\trans} (e', \cat{tl}{\gamma}) .
\end{align*}

Hence, $(\cC, \overline{\trans})$ and $(\TT, \step)$ are coalgebras of the same
functor. We can now build a map $f : \TT \rightarrow \cC$ from terms to
configurations:
\begin{align*} 
 f(\sigma) &\triangleq (\sigma, \varepsilon) \\
 f(\tl^k(\sigma)) &\triangleq (\sigma, tl^k) \\
 f(e_1 \oplus_p e_2) &\triangleq (e_1 \oplus_p e_2, \varepsilon) \\
 f(\tl^k(e_1 \oplus_p e_2)) &\triangleq (e_1 \oplus_p e_2, tl^k) \\
 f(\cons{\hat{u}}{e'}) &\triangleq (\cons{\hat{u}}{e'}, \varepsilon) \\
 f(\tl^k(\cons{\hat{u}}{e'})) &\triangleq (\cons{\hat{u}}{e'}, tl^k) .
\end{align*}
%
%
For every term $e$, we have
\[
  \overline{\trans}(f(e)) = {\rm case}(\step(e), \inl(a,e_1) \mapsto \inl(a, f(e_1)), \inr(e_2) \mapsto \inr(f(e_2)))
\]
Therefore, $f$ is a coalgebra homomorphism. By finality, this means that for all
$e \in \TT$, $\sem{e} = \paths(f(e))$, and so we conclude
\[
  \Pr_{t \sim \sem{e}} [ t \text{ has infinitely many output nodes} ]
  = \Pr_{\pi \sim {\paths}(f(e))}[\pi \models \always\finally\: \cO] .
\]
\end{proof}
\fi

By Theorem~\ref{thm:mc-dec}, ASP is decidable for stream terms. In fact, it is
also possible to decide whether a stream term is almost surely \emph{not}
productive, i.e., the probability of producing infinitely many outputs is zero.

 
\subsection{Extending to trees}
Now, we extend our approach to trees. The main difficulty can be seen in the
constructors. For streams, we can encode the term $\cons{a}{e}$ by proceeding to
the tail $e$. For trees, however, how can we encode $\mkt(a, e_1, e_2)$? The
pPDA cannot step to both $e_1$ and $e_2$. Since the failure of ASP may occur
down either path, we cannot directly translate the ASP property on trees to
LTL---ASP is a property of \emph{all} paths down the tree. Instead, on
constructors our pPDA encoding will choose a path at random to simulate. As we
will show, if the probability of choosing a path that outputs infinitely often
is 1, then every path will output infinitely often. Notice that in general,
properties that happen with probability 1 do not necessarily happen for every
path, but the structure of our problem allows us to make this generalization.

More formally, the stack alphabet will now be $\{rt, lt\}$, and on constructors
we transition to each child with probability $1/2$:
\[
\begin{aligned}
  \trans_e((\tau, a), (e, a)) &= 1 \\
  \trans_e((e_1 \oplus_p e_2, a),(e_1, a)) &= p \\
  \trans_e((e_1 \oplus_p e_2, a),(e_2, a)) &= 1 - p \\
  \trans_e((\mkt(a', e_l, e_r), \bot),(e_l,\varepsilon)) &= 1/2 \\
  \trans_e((\mkt(a', e_l, e_r), \bot),(e_r,\varepsilon)) &= 1/2 \\
\end{aligned}
\qquad
\begin{aligned}
  \trans_e((\mkt(a', e_l, e_r), lt),(e_l,\varepsilon)) &= 1 \\
  \trans_e((\mkt(a', e_l, e_r), rt),(e_r,\varepsilon)) &= 1 \\
  \trans_e((\lt(e'), a),(e', \cat{lt}{a})) &= 1 \\
  \trans_e((\rt(e'), a),(e', \cat{rt}{a})) &= 1
\end{aligned}
\]
We define $\cO \triangleq \{ s \in \cC \:|\: \exists a',e_l,e_r.\, s =
(\mkt(a',e_l,e_r), \varepsilon)\}$ to be the set of outputting configurations,
and we can characterize ASP with the following theorem.

\begin{theorem}
  Let $e$ be a tree term and $\cA_{e}$ be the corresponding probabilistic PDA. Then
  $
    \Pr_{\pi \sim \paths(e, \bot)}[\pi \models \always\finally\: \cO] = 1
  $
  if and only if for every $w \in \{L,R\}^\omega$,
  \[
    \Pr_{t \sim \sem{e}} [t \text{ has infinitely many output nodes along } w] = 1 .
  \]
\end{theorem}

Thus we can decide ASP by deciding a LTL formula. 

\iffull
\begin{proof}
The main result we need to prove is that given a distribution $\mu$ over $\OT$
and the distribution $\mu'$ over $\OS$ induced by $\mu$,
\[
  \Pr_{\pi \sim \mu'} [\pi \models \always\finally\: \cO] = 1
  \iff
  \forall w \in \{L,R\}^\omega.\,
  \Pr_{t \sim \mu} [\pi \text{ has infinitely many output nodes along } w ] = 1 .
\]
After this, all that remains is to check that the distribution over the runs of $\cA_e$
starting on $(e, \varepsilon)$ is exactly $\mu'$, which is done using similar techniques
as in the proof of \ref{thm:ASP-streams-dec}.

We start by showing how to compute this induced distribution. Let $F : X
\mapsto A \times X + X$ and $G : X
\mapsto A \times X \times X + X$ be the functors that generate $\OS$
and $\OT$ respectively. We define a
natural transformation $G \stackrel{\rho}{\Rightarrow} \Dist F$, which will
allow us to transform $G$-coalgebras into $\Dist F$-coalgebras, and in
particular $\OT$ into $\Dist(\OS)$.
We assign to every object $X$ a morphism $\rho_X: GX \to \Dist FX$ as follows:
   \begin{align*}
     \rho_X &: A \times X \times X + X \to \Dist(A \times X + X)\\
     \rho_X(\inr(x))  &= \delta(\inr(x))\\
     \rho_X(\inl(a, x, y))  &= 1/2 \cdot \delta(\inl(a,x)) + 1/2 \cdot \delta(\inl(a,y))
   \end{align*}
This gives us a map $f$ from $\OT$ to $\Dist(\OS)$ as the unique coalgebra
homomorphism closing the following commutative diagram (in the Kleisli category):

\begin{tikzcd}[column sep = 3cm, row sep = 2cm]
  \OT \ar[r, "f", "\circ" anchor=center, dashrightarrow]
      \ar[d, "\rho_{\OT} \circ \langle \out \text{,} \unf \rangle^{-1}", "\circ" anchor=center]
  & \OS \ar[d, "\langle \hd \text{,} \tl \rangle", "\circ" anchor=center] \\
  F(\OT) \ar[r, "Ff", "\circ" anchor=center, dashrightarrow]
  & F(\OS)  
\end{tikzcd}

We can see that, by uniqueness:
   \begin{align*}
     f &: \OT \to \Dist(\OS)\\
     f\: \unf(x)  &= \Dist(\unf(\bullet)) (f(x))\\
     f\: \out(a, x, y)  &= 1/2 \cdot \Dist(\cons{a}{\bullet})(f(x)) + 1/2 \cdot \Dist(\cons{a}{\bullet})(f(y)) .
   \end{align*}
Since $\Dist$ is a monad, we can extend $\rho$ to a natural transformation
$\Dist G \stackrel{\bar{\rho}}{\Rightarrow} \Dist F$ assigning to every $X$ a
morphism $\bar{\rho}_X = m_{FX} \circ \Dist\rho_X$, where $m$ is the product
of the $\Dist$ monad. Given a final $\Dist G$-coalgebra $(\Dist(\OT), h)$, we
have a $\Dist F$-algebra $(\Dist(\OT), \rho_{\Dist(\OT)}\circ h)$ so there is a
unique $\Dist F$-homomorphism $\hat{f}$ to the final $\Dist F$-coalgebra $(\Dist(\OS),
g)$. This allows us to give semantics in $D(OS)$ to a tree term:

\begin{tikzcd}[column sep = large, row sep = large]
  \OT \ar[d, "h", "\circ" anchor=center]
             \ar[rrr, "\hat{f}", "\circ" anchor=center, bend left = 30, dashrightarrow] 
  & \TT \ar[l, "\sem{-}_\OT", "\circ" anchor=center, dashrightarrow ]
        \ar[r, equal ]
        \ar[d, "\step", "\circ" anchor=center]
  & \TT \ar[d, "\bar{\rho}_\TT \circ \step", "\circ" anchor=center]
        \ar[r, "\sem{-}_\OS", "\circ" anchor=center, dashrightarrow ]
  & \OS \ar[d, "g", "\circ" anchor=center] \\
    G(\OT) \ar[d, "\bar{\rho}_{(\OT)}", "\circ" anchor=center]
  & G(\TT) \ar[l, "\circ" anchor=center, dashrightarrow]
                 \ar[r, "\bar{\rho}_\TT", "\circ" anchor=center ]
  & F(\TT) \ar[r, "\circ" anchor=center, dashrightarrow ]
  & F(\OS) \\
    F(\OT) \ar[urrr, bend right = 10, dashrightarrow, "F\hat{f}", "\circ" anchor=center]
  &
  &
\end{tikzcd}

Notice that, by uniqueness:
  \[
    \hat{f}(M)(E) = (m_{\OS} \circ (\Dist f))(M)(E) =
    \int_{t\in\OT} \left(  \int_{\pi\in\OS} \chi_E(\pi) d f(t) \right) d M
  \]
  where $\chi_E$ is the characteristic function of $E \subseteq \OS$.
  Now, let
  \begin{align*}
    S &= \{ \pi \:|\: \pi \models \always\finally\: \cO \}\\
    P_\pi &= \{ t \:|\: t_\pi \in S \}
  \end{align*}
  where $t_\pi$ for $\pi \in \{L,R\}^\omega$ is the path in $t$ corresponding to
  the choices $\pi$:
  \begin{align*}
    \unf(x)_{L:w} &= \unf(x_{L:w}) \\
    \unf(x)_{R:w} &= \unf(x_{L:w}) \\
    \out(a, x, y)_{L:w} &= \cons{a}{x_{w}} \\
    \out(a, x, y)_{R:w} &= \cons{a}{y_{w}}
  \end{align*} 
  Then
  \[\Pr_{\pi \sim \mu}[\pi \models \always\finally\: \cO] = \int_{\pi\in\OS} \chi_S(\pi) d\mu \]
  and
  \[\Pr_{t \sim \mu} [ t_\pi \text{ has infinitely many output nodes} ] = \int_{t\in\OT} \chi_{P_\pi}(t)d\mu . \]

  The rest of the proof proceeds in three steps. In the following, let $U$ be
  the distribution on $\{L,R\}^\omega$ assigning probability $(1/2)^k$ to every
  cone generated by a prefix of length $k$.
  First we need the following lemma stating that the distribution induced by $f$
  on $\OS$ is the same as the distribution induced by taking paths sampled from $U$.
  Intuitively, we are just pre-sampling the randomness in $f$ from $U$:
  \begin{lemma}\label{lem:trees-distr-equiv}
    Let $t \in \OT$. Then
    \[
      \int_{w \in \{L,R\}^\omega} \chi_{P_w}(t)dU = \int_{\pi \in \OS} \chi_S(\pi)df(t) .
    \]
  \end{lemma}
  \begin{subproof}[Proof of Lemma~\ref{lem:trees-distr-equiv}]
    We show that for every measurable $B \subseteq \OS$,
    \[
      \int_{w \in \{L,R\}^\omega } \chi_{B}(t_w) dU
      = \int_{\pi \in \OS} \chi_{B}(\pi) df(t) .
    \]
    For any distribution $N \in \Dist(\{L,R\}^\omega)$, there is an induced
    distribution by $t$ and $N$ on $\Dist(\OS)$:
    \begin{align*} 
      g &: \OT \to \Dist(\{L,R\}^\omega) \to \Dist(\OS)\\
      g\: \unf(x)\: N  &= \Dist(\unf(\bullet)) (g\ x\ N) \\
      g\: \out(n, x, y)\: N  &= \Pr_{w \sim N}[w[0] = L] \cdot \Dist(\cons{n}{\bullet})(g\: x\: (N\mid_{\tl}^L) )
      + \Pr_{w \sim N}[w[0] = R] \cdot \Dist(\cons{n}{\bullet})(g\: y\: (N\mid_{\tl}^R) )
    \end{align*}
    where $N\mid_{\tl}^X$ is the distribution on the tails of $N$ conditioned to the head being $X$.
    (if it is empty, we just make that side of the sum 0)
    What we are doing is taking from the first position of a $w$ sampled from $N$ the randomness
    for deciding which branch of the tree to take.
    In particular, it is easy to see that
    $g\ t\ U = f(t)$. Therefore,
    \[
      \int_{w \in \{L,R\}^\omega } \chi_{B}(t_w) dU
      = \int_{\pi \in \OS} \chi_{B}(\pi)d(g\ t\ U)
      = \int_{\pi \in \OS} \chi_{B}(\pi) df(t) .
    \]
  \end{subproof}

  Then following lemma shows that a tree produces paths with infinitely many outputs along every $w$
  in $\{L,R\}^\omega$ with probability 1 if and only if it produces paths with
  infinitely many outputs
  along a $w$ sampled from $\{L,R\}^\omega$ with probability 1. In other words, it
  provides a connection between the universal quantification in the definition
  of ASP and the probabilistic nature of LTL.
  \begin{lemma}\label{lem:trees-asp-equiv}
    Let $M \in \Dist(\OT)$. Then, the following are equivalent:
    \begin{enumerate}
      \item For every $w \in \{L,R\}^\omega$,  $\int_{\OT} \chi_{P_w}dM = 1$.
      \item $\int_{w \in \{L, R\}^\omega} \int_{\OT} \chi_{P_w} dM dU = 1$.
    \end{enumerate}
  \end{lemma}

  \begin{subproof}[Proof of Lemma~\ref{lem:trees-asp-equiv}]
    (1 $\Rightarrow$ 2) is immediate. To prove (2 $\Rightarrow$ 1),
    we suppose that there is a $w \in \{L,R\}^\omega$ such that $\int_{\OT}
    \chi_{P_w}dM < 1$, and we show that this must also be true for a $W \subseteq
    \{L,R\}^\omega$ such that $\int_{v \in \{L,R\}^\omega} \chi_W dU > 0$, and therefore
    $\int_{v \in \{L, R\}^\omega} \int_{\OT} \chi_{P_v} dM dU < 1$.

  To do this, we
  consider the set $\{w_n\}_{n\in\NN}$ of prefixes of $w$ of length $n$, and the cones
  $\{C_n\}_{n\in\NN}$ generated by them. For each of those prefixes $w_i$, we can compute the set
  of observation trees $T_i$ such that every $t \in T_i$ has $i$ output nodes along
  $w_i$. This set is measurable (the union of cones of all finite trees of height $i$
  satisfying the conditions),
  $T_1 \supseteq T_2 \supseteq T_3 \supseteq \dots$, and therefore:
  \[ 1 \geq \int_{T_1}dM \geq \int_{T_2}dM \geq \int_{T_3}dM \dots \]
  But since $P_w = \cap_{n} T_n$, the limit of this sequence is exactly $\int_{\OT}
  \chi_{P_w}dM < 1 $. Therefore $\int_{T_k}dM < 1$ for some $k$, and so
  \begin{align*}
    \int_{v \in C_k} \int_{\OT} \chi_{P_v} dM dU &= \int_{v \in C_k}
    \int_{T_k} \chi_{P_v} dM dU + \int_{v \in C_k} \int_{T_k^c} \chi_{P_v} dM dU = \\
    &= \int_{v \in C_k}\int_{T_k} \chi_{P_v} dM dU + 0 < \int_{v \in C_k} dU .
  \end{align*}
  Hence, we conclude
  \begin{align*}
    \int_{v \in \{L,R\}^\omega} \int_{\OT} \chi_{P_v} dM dU &=
    \int_{v \in C_k} \int_{\OT} \chi_{P_v} dM dU + \int_{v \in C_k^c} \int_{\OT} \chi_{P_v} dM dU \\
    &< \int_{v \in C_k}  dU + \int_{v \in C_k^c}  dU = 1 .
  \end{align*}

  \end{subproof}


   
  Finally we show that given $M \in \Dist(\OT)$, $\hat{f}(M)$
  produces paths with infinitely many outputs with probability 1 if and only if $\Dist(\OT)$ is ASP.
  \begin{lemma}\label{lem:int-equiv}
    Let $M \in \Dist(\OT)$. Then,
    \[
      \int_{\OT} \int_{\OS} \chi_S d f(t) dM = 1
      \iff
      \forall  w \in \{L,R\}^\omega.\,  \int_{\OT} \chi_{P_w}dM = 1 .
    \]
  \end{lemma}
  This is immediate by Lemma~\ref{lem:trees-asp-equiv}, Fubini's theorem
  and Lemma~\ref{lem:trees-distr-equiv}.
  


\end{proof}
\fi

\section{Possible Generalizations and Extensions} \label{sec:discussion}

Our definition of ASP and our verification approaches suggest several natural
directions.

\subparagraph*{Handling Richer Languages}
The most concrete direction is to consider richer languages for coinductive
probabilistic programming. Starting from our core language, one might consider
allowing more operations on coinductive terms, mutually recursive definitions,
or conditional tests of some kind.  It should also be possible to develop
languages for more complex coinductive types associated with general polynomial
functors (see, e.g., Kozen~\cite{DBLP:journals/entcs/Kozen11}). Note that adding
more operations, e.g. pointwise $+$ of streams would increase the expressivity
of the language but raise additional challenges from the perspective of the
semantics---we would have to add extra structure to the base category and
re-check that the finality proof. 

Developing new languages for coinductive probabilistic programming---perhaps an
imperative language or a higher-order language---would also be interesting.
From the semantics side, our development in Section~\ref{sec:def-asp} should
support any language equipped with a small-step semantics producing output
values, allowing ASP to be defined for many kinds of languages. The verification
side appears more challenging. Natural extensions, like a pointwise addition
operation, already seem to pose challenges for the analyses. We know of no
general method to reasoning about ASP. This stands in contrast to almost sure
termination, which can be established by where flexible criteria like decreasing
probabilistic variants~\cite{DBLP:journals/toplas/HartSP83}.  Considering
counterparts of these methods for ASP is an interesting avenue of research.

\subparagraph*{Exploring Other Definitions}
Our definition of ASP is natural, but other definitions are
possible. For trees (and possibly more complex coinductive
structures), we could instead require that there \emph{exists} a path
producing infinitely many outputs, rather than requiring that \emph{all}
paths produce infinitely many outputs. This weaker notion of ASP can
be defined in our semantics, but it is currently unclear how
to verify this kind of ASP.

Our notion of ASP also describes just the probability of generating
infinitely many outputs, and does not impose any requirement on the
generation rate. Quantitative strengthenings of ASP---say, requiring
bounds on the expected number of steps between outputs---could give
more useful information.

\subparagraph*{Understanding Dependence on Step Relation}
Our coalgebraic semantics supporting our verification methods are
based on a small-step semantics for programs. A natural question is
whether this dependence is necessary, or if one could verify ASP with a less
step-dependent semantics.  Again drawing
an analogy, it appears that fixing a reduction strategy is important
in order to give a well-defined notion of almost sure termination for
probabilistic higher-order languages~(see, e.g.,
\cite{DBLP:conf/esop/LagoG17}). The situation for almost sure productivity is
less clear.

\section{Related Work} \label{sec:rw}
Our work is inspired by two previously independent lines of research:
probabilistic termination and productivity of coalgebraic definitions.

\subparagraph*{Probabilistic Termination}
There are a broad range of techniques for proving termination of probabilistic
programs. Many of the most powerful criteria use advanced tools from probability
theory~\cite{DBLP:journals/pacmpl/McIverMKK18}, especially martingale
theory~\cite{Chakarov-martingale,DBLP:conf/popl/FioritiH15,Chatterjee2016,Chatterjee:2016:AAQ:2837614.2837639,Chatterjee:2017:SIP:3009837.3009873}.
Other works adopt more pragmatic approaches, generally with the goal of
achieving automation. Arons, Pnueli and Zuck~\cite{DBLP:conf/fossacs/AronsPZ03}
reduce almost sure termination of a program $P$ to termination of a
non-deterministic program $Q$, using a planner that must be produced by the
verifier. Esparza, Gaiser and Kiefer~\cite{DBLP:conf/cav/EsparzaGK12} give a
CEGAR-like approach for building patterns (which play a role similar to
planners) and prove that their approach is complete for a natural class of
programs.

\subparagraph*{Productivity of Corecursive Definitions}
There has been a significant amount of work on verifying productivity of
corecursive definitions without probabilistic choice.  Endrullis and
collaborators~\cite{DBLP:journals/tcs/EndrullisGHIK10} give a procedure for
deciding productivity of an expressive class of stream definitions. In a
companion work~\cite{DBLP:conf/lpar/EndrullisGH08}, they study the strength of
data oblivious criteria, i.e., criteria that do not depend on values. More
recently, Komendantskaya and
collaborators~\cite{DBLP:conf/lopstr/KomendantskayaJ16} introduce observational
productivity and give a semi-decision procedure for logic programs.

\section{Conclusion} \label{sec:conclusion}
We introduce almost sure productivity, a counterpart to almost sure
termination for probabilistic coinductive programs. In addition, we
propose two methods for proving ASP for a core language for streams
and infinite trees. Our results demonstrate that verification of ASP
is feasible and can even be decidable for simple languages. 


\subparagraph*{Acknowledgements.}

We thank Benjamin Kaminski, Charles Grellois and Ugo dal Lago for
helpful discussions on preliminary versions of this work, and we are
grateful to the anonymous reviewers for pointing out areas where the
exposition could be improved. This work was initiated at the Bellairs
Research Institute in Barbados, and was partially supported by NSF
grant
\#1637532 and ERC grant ProFoundNet (\#679127).



\bibliography{header,refs}

\end{document}